%% file: arxiv.tex
\documentclass{article}
\usepackage[utf8]{inputenc}
\usepackage[T1]{fontenc}
\usepackage[dvipsnames]{xcolor}
\usepackage{url}
\usepackage{booktabs}
\usepackage{amsfonts}
\usepackage{nicefrac}
\usepackage{microtype}
\usepackage{tcolorbox}
\usepackage{enumitem}
\usepackage{graphicx}
\usepackage{algorithm}
\usepackage{algorithmic}
\usepackage[switch]{lineno}
\usepackage{dirtytalk}
\usepackage{mathtools}
\usepackage{amsthm}
\usepackage{authblk}
\usepackage{natbib}
\usepackage[verbose=true,letterpaper]{geometry}
\AtBeginDocument{
  \newgeometry{
    textheight=9in,
    textwidth=6.5in,
    top=1in,
    headheight=14pt,
    headsep=25pt,
    footskip=30pt
  }
}
\usepackage[colorlinks, citecolor=black, linkcolor=black]{hyperref}
\usepackage[nameinlink]{cleveref}
\newtheorem{theorem}{Theorem}
\newtheorem{lemma}[theorem]{Lemma}
\newtheorem{definition}[theorem]{Definition}
\newtheorem{observation}[theorem]{Observation}
\newtheorem{corollary}[theorem]{Corollary}

\newtheorem{remark}[theorem]{Remark}

\newcounter{boxalgo}
\makeatletter
\newenvironment{boxalgo}[1]%
{\refstepcounter{boxalgo}%
\protected@edef\@currentlabelname{#1}%
\begin{tcolorbox}[title={Algorithm \arabic{boxalgo}: #1.}]\begin{tabular}{l}}{\end{tabular}\end{tcolorbox}}
\makeatother

\Crefname{boxalgo}{Algorithm}{Algorithms}

\makeatletter
\newenvironment{proofof}[1][\proofname] {\par\pushQED{\qed}\normalfont\topsep6\p@\@plus6\p@\relax\trivlist\item[\hskip\labelsep\bfseries#1\@addpunct{.}]\ignorespaces}{\popQED\endtrivlist\@endpefalse}
\makeatother

\input{math_commands}

\title{A Descent-based method on the Duality Gap\\ for solving zero-sum games}

\author[1]{Michail Fasoulakis}
\author[2,3]{Evangelos Markakis}
\author[4]{Giorgos Roussakis}
\author[2,3]{Christodoulos Santorinaios}

\affil[1]{Royal Holloway, University of London, UK}
\affil[2]{Archimedes/Athena RC, Greece}
\affil[3]{Athens University of Economics and Business, Greece}
\affil[4]{Foundation for Research and Technology--Hellas}

\date{}

\begin{document}
\maketitle

\begin{abstract}
We focus on the design of algorithms for finding equilibria in 2-player zero-sum games. Although it is well known that such problems can be solved by a single linear program, there has been a surge of interest in recent years for simpler algorithms, motivated in part by applications in machine learning. Our work proposes such a method, inspired by the observation that the duality gap (a standard metric for evaluating convergence in min-max optimization problems) is a convex function for bilinear zero-sum games. To this end, we analyze a descent-based approach, variants of which have also been used as a subroutine in a series of algorithms for approximating Nash equilibria in general non-zero-sum games.    
In particular, we study a steepest descent approach, by finding the direction that minimises the directional derivative of the duality gap function.
Our main theoretical result is that the derived algorithms achieve a geometric decrease in the duality gap and improved complexity bounds until we reach an approximate equilibrium. Finally, we complement this with an experimental evaluation, which provides promising findings. Our algorithm is comparable with (and in some cases outperforms) some of the standard approaches for solving 0-sum games, such as OGDA (Optimistic Gradient Descent/Ascent), even with thousands of available strategies per player. 
\end{abstract}
\section{Introduction}\label{intro}

Our work focuses on the design of algorithms for finding Nash equilibria in 2-player bilinear zero-sum games. Zero-sum games have played a fundamental role both in game theory, being among the first classes of games formally studied, and in optimization, as it is easily seen that their equilibrium solutions correspond to solving a min-max optimization problem. 
Even further, solving zero-sum games is in fact equivalent to solving linear programs, as properly demonstrated in \cite{Adler13}. 

Despite the fact that a single linear program (and its dual) suffices to find a Nash equilibrium, there has been a surge of interest in recent years, for faster algorithms, motivated in part by applications in machine learning. One reason for this is that we may have very large games to solve, corresponding to LPs with thousands of variables and constraints. A second reason could be that e.g., in learning environments, the players may be using iterative algorithms that can only observe limited information, hence it would be impossible to run a single LP for the entire game. As an additional motivation, finding new algorithms for such a fundamental problem can provide insights that could be of further value and interest.

The above considerations have led to a variety of approaches and algorithms, spanning already a few decades of research. Some of the earlier works on this domain have focused purely on an optimization viewpoint. In parallel to this, significant attention has been drawn to learning-oriented algorithms, such as first-order methods. The latter class of algorithms performs gradient descent or ascent on the utility functions of the two players, and some of the proposed variants have been very successful in practice, such as the optimistic gradient and the extra gradient methods \cite{K76,P80}. Several works have focused on 
theoretical guarantees for their performance, and a standard metric used in the analysis is the duality gap. This is simply the sum of the regrets of the two players in a given profile, and therefore the goal often amounts to proving appropriate rates of decrease for the duality gap over the iterations of an algorithm.

Our work is motivated by the observation that the duality gap is a convex function for zero-sum games. This naturally gives rise to the suggestion that instead of performing gradient descent on the utility function of a player, which is not a convex function, we could apply a descent procedure directly on the duality gap. It is not straightforward that this can indeed be useful as it is not a priori clear that we can perform a descent step fast (i.e., finding the direction to move to). Nevertheless, it can form the basis for investigating new approaches for zero-sum games.

\subsection{Our Contributions}\label{contrib}
Motivated by the above discussion, we propose and analyze an optimization approach for finding approximate Nash equilibria in zero-sum games. 
Our algorithm is a descent-based method applied to the duality gap function, and is essentially an adaptation of a subroutine in the algorithms of \cite{TS08,DFSS17,DFM23} which are for general games, tailored to zero-sum games and with a different objective function. 
The method is applying a steepest descent approach, where we find in each step the direction that minimises the directional derivative of the duality gap function and move towards that. 
In \Cref{sec:alg} we provide the algorithm and our theoretical analysis. Our main result is that the derived algorithm achieves a geometric decrease in the duality gap until we reach an approximate equilibrium. This implies that the algorithm terminates after at most $O\Big( \frac{1}{\rho} \cdot \log\Big(\frac{1}{\delta} \Big ) \Big)$ iterations with a $\delta$-approximate equilibrium, where $\rho$ is a parameter, related to the computation of the directional derivative. We exhibit that the method can also be further customized and show that a different variant also converges after $O(\frac{1}{\sqrt{\delta}})$ iterations.

In \Cref{sec:experiments}, we complement our theoretical analysis with an experimental evaluation. 
Even though the method does need to solve a linear program in each iteration to find the desirable direction, this turns out to be of much smaller size on average (in terms of the number of constraints) than solving the linear program of the entire game. We compare our method against standard LP solvers, but also against state-of-the-art procedures for zero-sum games, such as Optimistic Gradient Descent-Ascent (OGDA).  
Our findings are promising and reveal that the running time is comparable to (and often outperforms) OGDA, even with thousands of strategies per player. We therefore conclude that the overall approach deserves further exploration, as there are also potential ways of accelerating its running time, discussed in \Cref{sec:experiments}. 

\subsection{Related Work}\label{related}

As already mentioned, conceptually, the works most related to ours are \cite{TS08,DFSS17,DFM23}. Although these papers do not consider zero-sum games, they do utilize a descent-based part as a starting point. The main differences with our work is that first of all, their descent is performed with respect to the maximum regret among the two players, whereas we use the duality gap function. Furthermore the descent phase is only a subroutine of their algorithms, since it does not suffice to establish guarantees for general games. Hence their focus is less on the decent phase itself and more on utilizing further procedures to produce approximate equilibria.

There is a plethora of algorithms for linear programming and zero-sum games, which is impossible to list here, but we comment on what we feel are most relevant. When focusing on optimization algorithms for zero-sum games,  \cite{hoda2010smoothing} use Nesterov's first order smoothing techniques to achieve an $\epsilon$-equilibrium in $O(1/\epsilon)$ iterations, with added benefits of simplicity and rather low computational cost per iteration. Following up on that work,  \cite{gilpin2012first} propose an iterated version of Nesterov's smoothing technique, which runs within $O(\frac{||A||}{\delta(A)}\cdot \ln(1/\epsilon)) $ iterations. However, while this is a significant improvement, the complexity depends on a condition measure $\delta(A)$, with A being the payoff matrix, not necessarily bounded by a constant. Another optimization approach that is relevant in spirit to ours is via the Nikaido-Isoda function \cite{NI55} and its variants. E.g., in \cite{RCJ19} they run a descent method on the Gradient NI function, which is convex for zero-sum games. We are not aware though of any direct connection to the duality gap function that we use here.

Apart from the optimization viewpoint, there has been great interest in designing faster learning algorithms for zero-sum games. Although this direction started already several decades ago, e.g. with the fictitious play algorithm \cite{B51,R51}, it has received significant attention more recently given the relevance to formulating GANs in deep learning \cite{GPMXWOCB14} and also other applications in machine learning. Some of the earlier and standard results in this area concern convergence {\it on average}. That is, it has been known that by using no-regret algorithms, such as the Multiplicative Weights Update (MWU) methods \cite{AHK12} the empirical average of the players' strategies over time converges to a Nash equilibrium in zero-sum games. Similarly, one could also utilize the so-called Gradient Descent/Ascent (GDA) algorithms.

Within the last decade, there has also been a great interest in algorithms attaining the more robust notion of {\it  last-iterate convergence}. This means that the strategy profile $(x_t,y_t)$, reached at iteration $t$, converges to the actual equilibrium as $ t \to \infty$. 
Negative results in \cite{BP18} and \cite{MPP18} exhibit that several no-regret algorithms such as many MWU as well as GDA variants, do not satisfy last-iterate convergence. Motivated by this, there has been a series of works on obtaining algorithms with provable last iterate convergence. The positive results that have been obtained for zero-sum games is that improved versions of Gradient Descent such as the Extra Gradient method \cite{K76} or the Optimistic Gradient method \cite{P80} attain last iterate convergence. In particular, \cite{daskalakis2018training} and \cite{DBLP:conf/aistats/LiangS19} show that the optimistic variant of GDA (referred to as OGDA) converges for zero-sum games. Analogously, OMWU (the optimistic version of MWU) also attains last iterate convergence, shown in \cite{Daskalakis2019LastIterateCZ} and further analyzed in \cite{Wei2021LinearLC}.
The rate of convergence of optimistic gradient methods in terms of the duality gap was studied in \cite{COZ22}, and was later improved to $O({1}/{t})$ in \cite{CZ23}. Further approaches with convergence guarantees have also been proposed, based on variations of the Mirror-Prox method \cite{FMPV22} as well as primal-dual hybrid gradient methods \cite{LY23}.
Finally, several of these methods have also been studied beyond zero-sum games, including among others \cite{GPD20}, where Optimistic Gradient is analyzed for more general games and \cite{DBLP:conf/aistats/DiakonikolasDJ21} where positive results are shown for a class of non-convex and non-concave problems. The picture however is more complex for general games with negative results also established in 
\cite{DSZ21}.

\section{Preliminaries}
\label{prelim}
We consider bilinear zero-sum games $(\mR,-\mR)$, with $n$ pure strategies per player, where $\mR$ is the payoff matrix of the row player. We assume $\mR\in [0,1]^{n \times n}$ without loss of generality\footnote{We can easily see that we can do scaling for any $\mR \in \mathbb{R}^{n \times n}$ s.t. $\mR\in [0,1]^{n \times n}$ keeping exactly the same Nash equilibria.}. We consider mixed strategies $\vx \in \Delta^{n-1}$ as a probability distribution (column vector) on the pure strategies of a player, with $\Delta^{n-1}$ be the $(n-1)$-dimensional simplex. We also denote by $\ve_i$ the distribution corresponding to a pure strategy $i$, with 1 in the index $i$ and zero elsewhere. A strategy profile is a pair $(\vx,\vy)$, where $\vx$ is the strategy of the row player and $\vy$ is the strategy of the column player. Under a profile $(\vx,\vy)$, the expected payoff of the row player is $\vx^\top\mR\vy$ and the expected payoff of the column player is $-\vx^\top\mR\vy$.

A pure strategy $i$ is a $\rho$-best-response strategy against $\vy$ for the row player, for $\rho \in [0,1]$, if and only if, $\ve^\top_i\mR \vy +\rho \geq \ve^\top_j\mR \vy$, for any $j$. Similarly, a pure strategy $j$ for the column player is a $\rho$-best-response strategy against some strategy $\vx$ of the row player if and only if $\vx^T\mR \ve_j \leq \vx^T\mR \ve_i + \rho$, for any $i$. 
Having these,
we define as $BR^{\rho}_r(\vy)$ the set of the $\rho$-best-response pure strategies of the row player against $\vy$ and as $BR^{\rho}_c(\vx)$ the set of the $\rho$-best-response pure strategies of the column player against $\vx$.
For $\rho=0$, we will use $BR_r(\vy)$ and $BR_c(\vx)$ for the best response sets.

\begin{definition}[Nash equilibrium \cite{N51,VN28}]
A strategy profile $(\vx^*,\vy^*)$ is a Nash equilibrium in the game $(\mR,-\mR)$, if and only if, for any $i,j,$
\begin{equation*}
v={\vx^*}^\top \mR\vy^* \geq \ve_i^\top \mR\vy^*, \text{ and, }
v={\vx^*}^\top \mR\vy^* \leq {\vx^*}^\top \mR\ve_j,
\end{equation*}
\end{definition}
\noindent where $v$ is the value of the row player (value of the game).

\begin{definition}[$\delta$-Nash equilibrium]
A strategy profile $(\vx,\vy)$ is a $\delta$-Nash equilibrium (in short, $\delta$-NE) in the game $(\mR,-\mR)$, with $\delta\in [0,1]$, if and only if, for any $i,j,$
\begin{equation*}
\vx^\top \mR\vy +\delta\geq \ve_i^\top \mR\vy, \text{ and, }
\vx^\top \mR\vy -\delta\leq \vx^\top \mR\ve_j.
\end{equation*}
\end{definition}

With these at hand, we can now define the regret functions of the players as follows.
\begin{definition}[Regret of a player]
For a game $(\mR, -\mR)$, the regret function $f_\mR: \Delta^{n-1} \times \Delta^{n-1}\rightarrow [0,1]$ of the row player under a strategy profile $(\vx,\vy)$ is
\[f_\mR(\vx,\vy) = \max_{i} \ve_i^\top \mR\vy -\vx^\top \mR\vy.\]
Similarly, for the column player the regret function is
\[
f_{-\mR}(\vx,\vy) = \max_j \vx^\top (-\mR)\ve_j +\vx^\top \mR\vy = -\min_j \vx^\top \mR\ve_j +\vx^\top \mR\vy.\]
\end{definition}

An important quantity for evaluating the performance or convergence of algorithms is the sum of the regrets, i.e., the function $V(\vx,\vy) = f_\mR(\vx,\vy) + f_{-\mR}(\vx,\vy) = \max_i \ve_i^\top \mR\vy-\min_j \vx^\top \mR\ve_j$. This is referred to in the bibliography as the \textit{duality gap} in the case of zero-sum games.

\subsection{Warmup: Duality Gap Properties}
Next, we present some known results about the \textit{duality gap} function $V(\vx,\vy)$ and its connection to Nash equilibria. For the sake of completeness, we provide proofs.

\begin{theorem}\label{thm:wu1}
The duality gap $V(\vx,\vy)$ is convex in its domain.
\end{theorem}
\begin{proof}
    Let $(\vx_1,\vy_1)$ and $(\vx_2,\vy_2)$ be two arbitrary different strategy profiles, $p\in (0,1)$ and $(\vx,\vy) = p\cdot (\vx_1,\vy_1) +(1-p)\cdot (\vx_2,\vy_2) = (p\cdot \vx_1 + (1-p)\cdot \vx_2,p\cdot \vy_1 + (1-p)\cdot \vy_2) $ be a convex combination of them. Then, we have
\begin{align*}
    V(\vx,\vy) &= V(p\cdot \vx_1 + (1-p)\cdot \vx_2,p\cdot \vy_1 + (1-p)\cdot \vy_2) \\
    &= \max_i\ve_i^\top\mR (p\cdot \vy_1 + (1-p)\cdot \vy_2) - \min_j (p\cdot \vx_1 + (1-p)\cdot \vx_2)^\top \mR\ve_j\\ 
    &\leq p \cdot \max_{i} \ve_i^\top\mR\vy_1 + (1-p)\cdot \max_{i} \ve_i^\top\mR\vy_2\\& -p\cdot \min_j  \vx_1^\top \mR \ve_j-(1-p)\cdot \min_j \vx_2^\top \mR\ve_j\\
    &= p \cdot \max_{i} \ve_i^\top \mR\vy_1 + (1-p)\cdot \max_{i} \ve_i^\top\mR\vy_2 -p\cdot \min_j \vx_1^\top \mR\ve_j
    -(1-p)\cdot \min_j \vx_2^\top \mR\ve_j\\
    &= p \cdot V(\vx_1,\vy_1) + (1-p)\cdot V(\vx_2,\vy_2),
\end{align*}
the first inequality holds by the convexity and concavity of the $\max$ and the $\min$ function, respectively.
\end{proof}

\begin{theorem}\label{thm:wu2}
A strategy profile $(\vx^*,\vy^*)$ is a Nash equilibrium of the game $(\mR, -\mR)$, if and only if, it is a (global) minimum\footnote{Note that the set of Nash equilibria in zero-sum games and the set of optimal solutions, minimizing the duality gap are convex and identical to each other.} of the function $V(\vx, \vy)$.
\end{theorem}
\begin{proof}
Let $(\vx^*,\vy^*)$ be a Nash equilibrium, then it holds $V(\vx^*,\vy^*) = 0$ by the definition of the NE, but since the values of $V(\vx, \vy) \in [0,2]$ this implies that $(\vx^*,\vy^*)$ is a global minimum of the function in its domain. Let now a strategy profile $(\vx, \vy)$ such that $V(\vx, \vy)=0 = f_\mR(\vx, \vy) +f_{-\mR}(\vx, \vy)$, this trivially implies that $f_\mR(\vx, \vy) = 0$ and $f_{-\mR}(\vx, \vy) = 0$ since $f_\mR,f_{-\mR} \in [0,1]$, thus we have that $(\vx, \vy)$ is a NE in the zero-sum game.\end{proof}

\noindent Similarly to the previous theorem, we also have the following.

\begin{theorem}\label{thm:wu3]}
Let $(\vx, \vy)$ be a strategy profile in a zero-sum game. If $V(\vx, \vy) \leq \delta$, then $(\vx, \vy)$ is a $\delta$-NE. 
\end{theorem}

\section{Descent-based Algorithms on the Duality Gap:\\ Theoretical Analysis}
\label{sec:alg}
In this section, we present our main algorithm along with some improved variants, based on a gradient-descent approach for the function $V(\vx,\vy)$ in zero-sum games. The algorithm can be seen as an adaptation\footnote{Here as the objective function we use the sum of the regrets instead of the maximum of the two regrets.} of a descent procedure that forms the initial phase of algorithms proposed for general non-zero-sum games, in \cite{TS08,DFSS17,DFM23}. 
The main idea behind the algorithm is that since the global minimum of the duality gap function $V(\vx, \vy)$ is a Nash equilibrium and the duality gap is a convex function for zero-sum bilinear games, we use a descent method based on the directional derivative of $V(\vx, \vy)$. This differs substantially from applying the more common idea of gradient descent/ascent (GDA) on the utility functions of the players, which are not convex functions. To identify the direction that minimizes the directional derivative at every step we use linear programming (albeit solving much smaller linear programs on average than the program describing the zero-sum game itself). 

To begin with, we define first the directional derivative. 
\begin{definition}\label{def:dirder}
The directional derivative of the duality gap at a point $\vz = (\vx, \vy)$, with respect to a direction $\vz^\prime = (\vx^\prime, \vy^\prime) \in \Delta^{n-1} \times \Delta^{n-1}$ is the limit, if it exists,
\[
\nabla_{\vz^\prime}V(\vz) = 
\lim_{\varepsilon \to 0} 
\frac{V\Big((1-\varepsilon) \cdot \vz +\varepsilon \cdot \vz^\prime\Big) 
  - V(\vz)}{\varepsilon}
\]
\end{definition}

We provide below a much more convenient form for the directional derivative that facilitates the remaining analysis.

\begin{theorem}
\label{thm:comb-version}
The directional derivative of the duality gap $V$ at a point $\vz = (\vx, \vy)$ with respect to a direction $\vz^\prime = (\vx^\prime,\vy^\prime) \in \Delta^{n-1} \times \Delta^{n-1} $, is given by
\[
\nabla_{\vz^\prime}V(\vz) = \max_{i\in BR_r(\vy)} \ve_i^\top \mR\vy^\prime-\min_{j\in BR_c(\vx )}(\vx^\prime)^\top \mR\ve_j - V(\vz)
\]
\end{theorem}

Furthermore, by the definition of directional derivative we have the following consequence.
\begin{lemma}
\label{lem:minimum directional derivative}
Given $\delta \in [0,1]$, let $\vz = (\vx, \vy)$ be a strategy profile that is not a $\delta$-Nash equilibrium. Then
\[
\nabla_{\vz^\prime}V(\vz) < -\delta,
\]
where $\vz^\prime = (\vx^\prime,\vy^\prime) \in \Delta^{n-1} \times \Delta^{n-1}$ is a direction that minimizes the directional derivative.
\end{lemma}

The proof of \Cref{lem:minimum directional derivative} follows by a more general result presented in \Cref{lem:bound-approx_DD} below (using also \Cref{lem:Bound approximation and exact derivative}). In a similar manner to \Cref{def:dirder}, we define now an approximate version of the directional derivative. The reason we do that will become clear later on, in order to show that the duality gap decreases from one iteration of the algorithm to the next. The main idea in the definition below is to include approximate best responses in the maximization and minimization terms involved in the statement of \Cref{thm:comb-version}. Namely, for $\rho\!>\!0$, recall the definition of $BR^{\rho}_r(\vy)$ as the set of $\rho$-best response strategies of the row player against strategy $\vy$ of the column player (and similarly
for $BR^{\rho}_c(\vx )$). 

\begin{definition}[$\rho$-directional derivative]
The $\rho$-directional derivative of the duality gap $V$ at a point $\vz = (\vx, \vy)$ with respect to a direction $\vz^\prime = (\vx^\prime,\vy^\prime)  \in \Delta^{n-1} \times \Delta^{n-1}$ is  
\[
\nabla_{\rho,\vz^\prime}V(\vz) = \max_{i\in BR^{\rho}_r(\vy)} \ve_i^\top \mR\vy^\prime-\min_{j\in BR^{\rho}_c(\vx )}(\vx^\prime)^\top \mR\ve_j - V(\vz).
\]
\end{definition}
 \begin{lemma}
 \label{lem:Bound approximation and exact derivative}
 It holds that for any direction $\vz' = (\vx^\prime,\vy^\prime) \in \Delta^{n-1} \times \Delta^{n-1}$,  and for any $\rho>0$, $$\nabla_{\vz^\prime}V(\vz) \leq \nabla_{\rho,\vz^\prime}V(\vz).$$
 \end{lemma}

\begin{lemma}
\label{lem:bound-approx_DD}
Given $\delta \in [0,1]$, let $\vz = (\vx, \vy)$ be a strategy profile that is not a $\delta$-Nash equilibrium. Then
\[
\nabla_{\rho,\vz^\prime}V(\vz) < -\delta,
\]
where $\vz^\prime = (\vx^\prime,\vy^\prime) \in \Delta^{n-1} \times \Delta^{n-1}$ is a direction that minimizes the $\rho$-directional derivative.\end{lemma}
\noindent The proofs of these lemmas and any other missing proof from this section are given in \Cref{app:missing-main}.

\subsection{The Main Algorithm}
We now present our algorithm. \Cref{alg:main} takes as input a game and 3 parameters, namely $\delta\in (0,1]$, which refers to the approximation guarantee that is desired, $\rho\in (0,1]$ which involves the approximation to the directional derivative, and $\epsilon$, which refers to the size of the step taken in each iteration. 
Our theoretical analysis will require $\rho$ and $\epsilon$ to be correlated. 

\begin{boxalgo}{The gradient descent-based algorithm}\label{alg:main}
\hphantom{O}INPUT: A game $(\mR, -\mR)$, an approximation parameter $\delta\in (0, 1]$ and constants $\rho, \varepsilon \in (0, 1]$.\\
OUTPUT: A $\delta$-NE strategy profile\\\\
\hphantom{OUTPUT:} Pick an arbitrary strategy profile $(\vx, \vy)$\\
\hphantom{OUTPUT:} While $V(\vx, \vy)> \delta$\\
\hphantom{OUTPUT:} \quad \quad $(\vx^\prime,\vy^\prime)= $ FindDirection$(\vx, \vy, \rho)$\\
\hphantom{OUTPUT:} \quad \quad $(\vx, \vy) = (1-\varepsilon)\cdot (\vx, \vy) + \varepsilon \cdot (\vx^\prime,\vy^\prime)$\\
\hphantom{OUTPUT: } Return $(\vx, \vy)$  
\end{boxalgo}

\begin{boxalgo}{FindDirection($\vx, \vy, \rho$)}\label{alg:find_dir}
\hphantom{O}INPUT: A strategy profile $(\vx, \vy)$ and parameter $\rho\in (0,1]$.\\
OUTPUT: The direction $(\vx^\prime,\vy^\prime)$ that minimizes the $\rho$-directional derivative.\\\\
\hphantom{OUTPUT: }Solve the linear program (w.r.t. $(\vx^\prime,\vy^\prime)$ and $\gamma$):\\
\hphantom{OUTPUT: }\quad \quad minimize  $\gamma$    \\ 
\hphantom{OUTPUT : \quad \quad mini} s.t.  $\gamma \geq (\ve_i)^\top \mR\vy^\prime -(\vx^\prime)^\top \mR\ve_j,$ \\ \hphantom{OUTPUT: \quad \quad minimize } $i \in BR^{\rho}_r(\vy), j \in BR^{\rho}_c(\vx) \text{ and } \vx^\prime,\vy^\prime \in \Delta^{n-1}$\\
\hphantom{OUTPUT: } Return $(\vx^\prime,\vy^\prime)$    
\end{boxalgo}

\begin{observation}\label{obs:rho=1}
If $\rho = 1$, then \Cref{alg:find_dir} returns an exact Nash equilibrium of the game $(\mR,-\mR)$.
\end{observation}
The proof of this simple observation is referred to \Cref{app:find_dir}. We conclude the presentation of our main algorithm with the following remark. 
\begin{remark} The choice of $\rho$ demonstrates the trade off between global optimization (Linear Programming) and the descent-based approach. In the extreme case where $\rho=1$, \Cref{obs:rho=1} shows one iteration would suffice, solving the (large) linear program of the entire zero-sum game. On the other hand, when $\rho$ is small, close to $0$, then the method solves in each iteration rather small linear programs in \Cref{alg:find_dir} (dependent on the sets $BR^{\rho}_c(\vx), BR^{\rho}_r(\vy)$). 
\end{remark}
\subsection{Proof of Correctness and Rate of Convergence}
Our main result is the following theorem.
\begin{theorem}
\label{thm:main}
For any constants $\delta, \rho \in (0, 1]$, and with $\epsilon = \rho/2$, \Cref{alg:main} returns a $\delta$-Nash equilibrium in bilinear zero-sum games after at most $O(\frac{1}{\rho \cdot \delta} \log \frac{1}{\delta})$ iterations, and with a geometric rate of convergence for the duality gap.    
\end{theorem}
In order to prove this theorem, we will start first with the following auxiliary lemma.
\begin{lemma}
\label{lem:zero}
If $\varepsilon \leq \frac{\rho}{2}$, then it holds that 
\[
\max \Big\{0, \max_{i\in \overline{BR^{\rho}_r(\vy)}} \ve_i^\top \mR\Big((1-\varepsilon)\cdot \vy +\varepsilon \cdot \vy^\prime\Big)-\max_{i\in BR^{\rho}_r(\vy)} \ve_i^\top \mR\Big((1-\varepsilon) \cdot \vy +\varepsilon \cdot  \vy^\prime\Big) \Big\} =0.
\]
Similarly, for the column player, it holds that
\[
\max \Big\{0, -\min_{j\in \overline{BR^{\rho}_c(\vx)}}\Big((1-\varepsilon)\cdot \vx +\varepsilon \cdot \vx^\prime\Big)^\top  R\ve_j +\min_{j\in BR^{\rho}_c(\vx)} \Big((1-\varepsilon)\cdot \vx +\varepsilon \cdot \vx^\prime\Big)^\top  R\ve_j \Big\}=0.
\]
\end{lemma}

We can now establish that the duality gap decreases geometrically, as long as we have not yet found a $\delta$-approximate equilibrium. We first show an additive decrease.

\begin{lemma}
\label{lem:decrease}
    Let $\epsilon\leq \frac{\rho}{2}$ and suppose that after $t$ iterations we are at a profile $(\vx^t, \vy^t)$, which is not a $\delta$-Nash equilibrium. Then,  
    \begin{equation*}
V(\vx^{t+1}, \vy^{t+1}) \leq  V(\vx^t, \vy^t) - \epsilon\cdot\delta
    \end{equation*}
where $(\vx^{t+1},\vy^{t+1})$ is the strategy profile at iteration $t+1$.
\end{lemma}

\begin{proof}
To shorten notation, let $\vx^t=\vx,\vy^t=\vy, \vx^{\prime t} = \vx^\prime$, $\vy^{\prime t}\!=\!\vy^\prime, \vz^t=(\vx, \vy),\vz^{t+1}=(\vx^{t+1}, \vy^{t+1})$. Then we have
$$(\vx^{t+1}, \vy^{t+1})  = ((1-\varepsilon)\cdot \vx + \varepsilon\cdot\vx^\prime, (1-\varepsilon)\cdot \vy + \varepsilon\cdot\vy^\prime ). $$
Similar to \Cref{Maximum analysis} in \Cref{app:dr-form}, we have that  
\[
\max_i \ve_i^\top \mR\vy^{t+1} = \max_{i\in BR^{\rho}_r(\vy)} \ve_i^\top \mR\vy^{t+1}
+ \max \Big\{0, \max_{i\in \overline{BR^{\rho}_r(\vy)}} \ve_i^\top \mR\vy^{t+1} - \max_{i\in BR^{\rho}_r(\vy)} \ve_i^\top \mR\vy^{t+1} \Big\}.
\]
Note that since $\varepsilon\leq \frac{\rho}{2}$, \Cref{lem:zero} applies and zeroes out the last term. Respectively, we obtain that \[\min_j (\vx^{t+1})^\top\mR\ve_j = \min_{j\in {BR^{\rho}_c(\vx)}} (\vx^{t+1})^\top\mR\ve_j\]
Hence, 
\begin{align*}
\!V(\vz^{t+1})
&= \max_i \ve_i^\top \mR\vy^{t+1} - \min_j (\vx^{t+1})^\top\mR\ve_j \\
&= \max_{i\in BR^{\rho}_r(\vy)} \ve_i^\top \mR\vy^{t+1} - \min_{j\in {BR^{\rho}_c(\vx)}} (\vx^{t+1})^\top\mR\ve_j \\
&= \hphantom{-} \max_{i\in {BR^{\rho}_r(\vy)}} \ve_i^\top \mR\Big((1-\varepsilon) \cdot \vy +\varepsilon \cdot \vy^\prime\Big)
-\min_{j\in {BR^{\rho}_c(\vx)}} \Big((1-\varepsilon)\cdot \vx +\varepsilon \cdot \vx^\prime\Big)^\top  \mR\ve_j \\
&\leq \hphantom{-}(1-\varepsilon) \max_{i} \ve_i^\top \mR \vy
+\varepsilon \max_{i\in {BR^{\rho}_r(\vy)}} \ve_i^\top \mR \vy^\prime -(1-\varepsilon) \min_{j} \vx^\top  \mR\ve_j - \varepsilon\!\min_{j\in {BR^{\rho}_c(\vx)}} ( \vx^\prime )^\top  \mR\ve_j\\
&= \max_{i} \ve_i^\top \mR \vy-\min_{j} \vx^\top  \mR\ve_j
 + \varepsilon \Big(\max_{i\in {BR^{\rho}_r(\vy)}} \ve_i^\top \mR \vy^\prime -  \min_{j\in {BR^{\rho}_c(\vx)}} ( \vx^\prime )^\top  \mR\ve_j
-\max_{i} \ve_i^\top \mR \vy + \min_{j} \vx^\top  \mR\ve_j \Big)\\
&=  V(\vz^t) + \varepsilon \cdot \nabla_{\rho,\vz^\prime}V(\vz^t) 
< V(\vz^t) - \varepsilon \cdot \delta,
\end{align*}
where the last inequality follows from \Cref{lem:bound-approx_DD}.
\end{proof}

The next step is to turn the additive decrease of \Cref{lem:decrease} into a multiplicative decrease.
\begin{corollary}
\label{cor:decrease}
For $\epsilon = \rho/2$, we have that
\begin{equation*}
   V(\vx^{t+1}, \vy^{t+1}) \leq \Big(1-\frac{\rho \cdot \delta}{4}\Big) \cdot  V(\vx^t, \vy^t)
    \end{equation*}
\end{corollary}
\begin{proof}
Using \Cref{lem:decrease}, we get that 
$V(\vz^{t+1}) \le (1-c)\cdot V(\vz^t)$
with $c = \frac{\varepsilon \cdot \delta}{V(\vz^t)} \geq  \frac{\rho \cdot \delta}{4} $, since $V(\vx, \vy) \leq 2$ for any profile, and $\varepsilon = \frac{\rho}{2}$. 
\end{proof}

Finally, we can complete the proof of our main theorem.

\begin{proofof}[Proof of {\hypersetup{linkcolor=black}\Cref{thm:main}}]
We have already proved the geometric decrease of the duality gap, for constant $\rho$ and $\delta$. Hence, the algorithm eventually will satisfy that the duality gap is at most $\delta$ and will terminate with a $\delta$-NE. It remains to bound the number of iterations that are needed. 
Suppose that the algorithm terminates after $t$ iterations, with profile 
$(\vx^t,\vy^t)$. 
By repeatedly applying \Cref{cor:decrease}, we have that
\[ V(\vx^t,\vy^t) \le  (1-c)^t \cdot V(\vx^0,\vy^0) \]
with $c = \frac{\rho\cdot \delta}{4}$.
In order to ensure that $V(\vx^t, \vy^t)\leq \delta$, it suffices to have that $2 \cdot (1-c)^t \leq \delta$, since $V(\vx^0,\vy^0) \leq 2$.
\[
2(1-c)^t \le \delta \implies t \ge \frac{\log \frac{2}{\delta}}{\log\frac{1}{1-c}} \implies t\ge  \frac{1-c}{c}\log \frac{2}{\delta}
\]
where the last inequality holds due to $\log x \le x-1$, for $x\geq 1$. Since $\frac{1-c}{c} =O(\frac{1}{c})$, the proof is completed by substituting the value of $c$. \end{proofof}

\subsection{Decaying Schedule Speedups}

In this section, we present a different implementation of our main approach, which results in an improved analysis. The idea is to gradually decay $\delta$ and use it to bound $c$, instead of the more coarse approximation of $V(\vx, \vy) \le 2$, that we used in the proof of \Cref{thm:main}.
This is presented as \Cref{alg:d_speed}.

\begin{boxalgo}{Decaying Delta Speedup}\label{alg:d_speed}
\hphantom{O}INPUT: A game $(\mR, -\mR)$, an approximation parameter $\delta \in (0,1]$ and a constant $\rho\in (0, 1]$.\\
OUTPUT: A $\delta$-NE strategy profile.\\\\
\hphantom{OUTPUT:} Pick an arbitrary strategy profile $(\vx, \vy)$\\
\hphantom{OUTPUT:} Set $i = 0,\delta_0 = 1$ and $ \varepsilon = \frac{\rho}{2}$\\
\hphantom{OUTPUT:} While TRUE\\
\hphantom{OUTPUT:} \quad \quad Set $i = i+1$ and $\delta_i = \delta_{i-1} / 2$\\
\hphantom{OUTPUT:} \quad \quad Update $(\vx, \vy)$ via Algorithm 1 $\Big((\mR,-\mR), \delta_i, \rho, \varepsilon \Big)$\\
\hphantom{OUTPUT:} \quad \quad If $\delta_i \le \delta:$ break
\\
\hphantom{OUTPUT: } Return $(\vx, \vy)$\\    
\end{boxalgo}

\begin{theorem}\label{thm:decaydelta}
     \Cref{alg:d_speed} maintains a geometric decrease rate in the duality gap and reaches a $\delta$-NE after at most $O\Big( \frac{1}{\rho} \cdot \log\Big(\frac{1}{\delta} \Big ) \Big)$ iterations.
\end{theorem}
\begin{proof}
We think of the iterations of the entire algorithm as divided into epochs, where each epoch corresponds to a new value for $\delta$.
Fix an epoch $i$, with $i>0$. Within this epoch, \Cref{alg:main} is run with approximation parameter $\delta_i$. Consider an arbitrary iteration of \Cref{alg:main} during this epoch, say at time $t+1$, starting with the profile $\vz^t = (\vx^t, \vy^t)$ and ending at the profile $\vz^{t+1}=(\vx^{t+1}
, \vy^{t+1})$. By \Cref{lem:decrease}, we have that $V(\vz^{t+1}) \le V(\vz^t) - \epsilon\cdot \delta_i = (1-c_i)\cdot V(\vz^t)$, where
$c_i =  \frac{\epsilon \cdot \delta_i }{V(\vz^t)} = \frac{\rho \cdot \delta_i }{2 \cdot V(\vz^t)}$. Since we are at epoch $i$, we know that $V(\vz^t) \le \delta_{i-1} = 2 \cdot \delta_i$, because the duality gap was at most $\delta_{i-1}$ at the beginning of epoch $i$ and within the epoch it only decreases further due to \Cref{lem:decrease} (for epoch 1, it is even better, since $V(\vz^t) \leq V(\vz^0)\le 2 = 2\delta_0\leq 4\delta_1$, where $\vz^0$ is the initial profile). 
Therefore,
$c_i \geq  \frac{\rho \cdot \delta_i }{2 \cdot \delta_{i-1}}=\frac{\rho}{4}$. Hence, we have established that in any iteration, regardless of the epoch:
\begin{equation*}
\begin{split}
    V(\vz^{t+1}) \leq \Big(1-\frac{\rho}{4}\Big) \cdot  V(\vz^t) \leq  \Big(1-\frac{\rho}{4}\Big)^t \cdot  V(\vz^0).
\end{split}
\end{equation*}

Since $\rho$ is constant, we have a geometric decrease, and this proves the first part of the theorem.

To bound the total number of iterations, let $t_i$ be the number of iterations of \Cref{alg:main} within epoch $i$, after which, the algorithm achieves a $\delta_i$-NE. 
Then, similar to the proof of \Cref{thm:main}, and since in the beginning of epoch $i$, the duality gap is at most $\delta_{i-1}$, we have that $t_i$ should satisfy
\[(1-c_i)^{t_i} \cdot \delta_{i-1} \leq \delta_i
    \implies t_i \ge \frac{1}{\log\frac{1}{1-c_i}} \implies t_i\ge  \frac{1-c_i}{c_i}
\]
Thus, at epoch $i$, we need $t_i = O(\frac{1}{\rho})$ to reach a $\delta_i$-NE. Next, note that if $k$ is the total number of epochs required to achieve a $\delta$-NE, when starting with $\delta_0$, it holds that $\frac{\delta_0}{2^k} \le \delta \implies k \geq \log\frac{\delta_0}{\delta}$. Since $\delta_0= 1$, the number of required epochs is $O(\log\frac{1}{\delta})$. Therefore, the total number of iterations for the entire algorithm is $O\Big( \frac{1}{\rho} \cdot \log\frac{1}{\delta} \Big)$.\end{proof}
To demonstrate the flexibility of our approach, we conclude the theoretical exploration with yet another variation, where we additionally use a decreasing schedule for the value of $\rho$. Specifically, this gives rise to the following scheme which we refer to as \Cref{alg:dr_speed}.
    
\begin{itemize}
    \item Use the same schedule for $\delta_i$ as \Cref{alg:d_speed}. 
    \item At iteration $i$ set $\rho_i = \sqrt{\delta_i}$, for \Cref{alg:main} (with $\varepsilon_i = \frac{\rho_i}{2}$).
\end{itemize} 
Note that we have now eliminated the dependence on $\rho$ but at the expense of making more expensive the dependence on $\delta$.
\begin{boxalgo}{Decaying Delta and Rho Speedup}\label{alg:dr_speed}
\hphantom{O}INPUT: A game $(\mR, -\mR)$ and an approximation parameter $\delta \in (0, 1]$.\\
OUTPUT: A $\delta$-NE strategy profile.\\\\
\hphantom{OUTPUT:} Pick an arbitrary strategy profile $(\vx, \vy)$\\
\hphantom{OUTPUT:} Set $i = 0$ and $\delta_0 = 1$\\
\hphantom{OUTPUT:} While TRUE\\
\hphantom{OUTPUT:} \quad \quad Set $i = i+1$ and $\delta_i = \delta_{i-1} / 2$\\
\hphantom{OUTPUT:} \quad \quad Set $\rho_i = \sqrt{\delta_i}$ and $\varepsilon = \frac{\rho_i}{2}$\\
\hphantom{OUTPUT:} \quad \quad Update $(\vx, \vy)$ via Algorithm 1 $\Big((\mR,-\mR), \delta_i, \rho_i, \varepsilon \Big)$\\
\hphantom{OUTPUT:} \quad \quad If $\delta_i \le \delta:$ break\\
\hphantom{OUTPUT:} Return $(\vx, \vy)$\\
\end{boxalgo}

\begin{theorem}\label{thm:sqrtdelta}
     \Cref{alg:dr_speed} reaches a $\delta$-Nash equilibrium after at most $O\Big(\frac{1}{\sqrt{\delta}} \Big)$ iterations, for any constant $\delta$.
\end{theorem}
\begin{proof}
    Following the same analysis as in the proof of \Cref{thm:decaydelta}, we obtain that for any 2 consecutive iterations within epoch $i$, we have that $V(\vz^{t+1}) \le (1-c_i)\cdot V(\vz^t)$, where
$c_i \geq  \frac{\rho_i}{4}$. This implies that if $t_i$ is the number of iterations needed within epoch $i$, it holds that $t_i \leq \lceil \frac{4}{\rho_i} \rceil \leq \frac{4}{\sqrt{\delta_i}} + 1$. 

We also have again that the total number of epochs is 
$k = O(\log(1/\delta))$. Putting everything together, and since $\delta_i = \frac{1}{2^{i}}$, the total number of iterations is 
    \[ t = \sum_{i=1}^k t_i \leq \sum_{i=1}^k (\frac{4}{\sqrt{\delta_i}}+1) = k + 4\cdot \sum_{i=1}^k \frac{1}{\sqrt{\frac{1}{2^{i}}}} = k + O \Big( 2^{k/2}\Big) = O \Big( \frac{1}{\sqrt{\delta}} \Big)\]\end{proof}

\section{Experimental Evaluation}\label{sec:experiments}

All our algorithms were implemented in Python, and the exact specifications can be found in \Cref{app:exp}. Before proceeding to our main findings, we exhibit first that the geometric decrease in the duality gap can indeed be observed experimentally. \Cref{fig:duality} shows a typical behavior of our algorithms, in terms of the duality gap. The figure here is for a random game of size $n=1000$.

\begin{figure}[b]
    \centering
    \includegraphics[scale=.4]{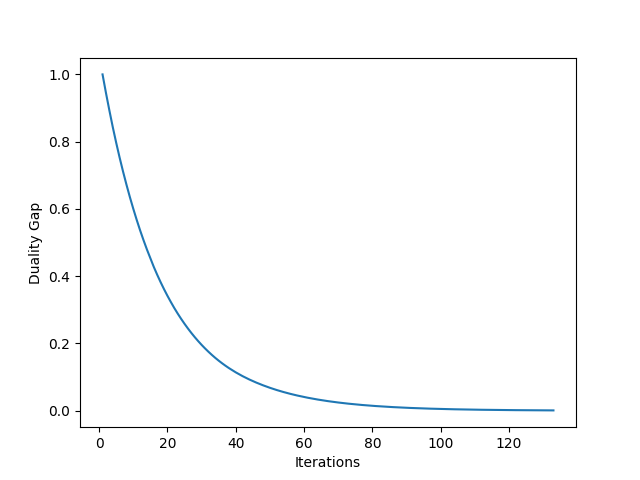}
    \caption{The decrease in the duality gap for a random game.}
    \label{fig:duality}
\end{figure}

\subsection{From Theory to Implementation} 
\label{sec:theory-to-impl}
We deem useful to discuss first how to approach the selection of the parameters that the algorithms depend on. We have seen in \Cref{alg:main} and its variants two families of parameters: $\rho_i$ and $\delta_i$. A third parameter is the learning rate $\varepsilon$, which is the step size that we take in each iteration.
\paragraph{Choice of $\bm\varepsilon$} We have established that as long as $\varepsilon \le \rho/2$, the points along the line $(1-\varepsilon)\cdot(\vx,\vy) + \varepsilon \cdot (\vx', \vy')$ decrease the duality gap (\Cref{lem:decrease}). Note, though, that the problem of minimizing $V$ along this set is a convex optimization problem. Hence, we can try to find the optimal $\varepsilon_i$ at each iteration $i$, and there are a few possible approaches for this: line search, ternary search or even solving it exactly using dynamic programming. We decided to use the following heuristic: for large values of the duality gap, namely $V>0.1$, we employ ternary search and as the duality gap decreases we use line search but only on a small part of the line. More specifically, once $V\le 0.1$ we start with $\varepsilon = 0.2$ and decrease it by $10\%$ across iterations. We decided upon this method since we noticed that experiments conform to theory for smaller values of $V$ and $\rho$. Finally, a more ML-like approach would be to set a constant $\varepsilon$, similarly to a constant step size $\eta$ in gradient methods. While this approach has merit, it did not show improved performance (see more in \Cref{app:eta}).  

\paragraph{Choice of $\bm\rho$ (and a new algorithm)} The most critical parameter regarding the running time of our algorithms is $\rho$, since it controls the size of the LPs in \Cref{alg:find_dir}, i.e., the number of constraints, via the sets of $\rho$-approximate best responses, $BR^{\rho}_r(\vy)$ and $BR^{\rho}_c(\vx)$. We need $\rho$ to be large enough to avoid having only a single best response, in which case our algorithms reduces to Best Response Dynamics, while at the same time it should be small enough so that the LPs have small size and we can solve them fast. Our experimentation did not reveal any particular range of $\rho$ with a consistently better performance. As a result, in addition to our existing algorithms, we developed one more approach, independent of $\rho$: we fix a number $k$ (much smaller than $n$), and in every iteration, we include in the approximate best response set of each player its top $k$ better responses. We refer to this approach as the \textit{Fixed Support Variant} in the sequel. We used $k\!=\!100$ for our experiments and point to \Cref{app:k} for justification. 

\paragraph{Optimizing FindDirection} For this algorithm we use two implementation tricks. The first one is quite simple: it is easy to observe that the LP of \Cref{alg:find_dir} is equivalent to solving two smaller LPs (see \Cref{app:find_dir}); it turns out that solving it this way is faster. The second trick revolves around $\rho$. Recall that the direction we find is itself an approximation. Hence, solving the LP approximately is meaningful, in the sense that it provides an even coarser approximate direction. It turns out that even a $0.1$ approximate solution (which is achievable by setting an appropriate parameter in the LP solver that we used) works for most cases, 
and results in significantly less running time (see also the discussion in \Cref{app:low-rank}).
\subsection{Comparisons between Our Variants}
\label{sec:exp-ours}
We report  first on our comparisons between \Cref{alg:d_speed} with $\rho = 0.001$, henceforth called the \textit{Constant $\rho$ Variant}, \Cref{alg:dr_speed} with $\rho_i = 0.01\sqrt{\delta_i}$, which we refer to as the \textit{Adaptive $\rho$ Variant} and our Fixed Support Variant discussed in \Cref{sec:theory-to-impl}.
We note that for the variant with the adaptive value of $\rho$, we did not follow precisely the values presented by our theoretical analysis, of $\rho_i = \sqrt{\delta_i}$. Although theoretically equivalent, this change was only  to avoid a blowup in the number of best response strategies used in \Cref{alg:find_dir} during the first iterations, i.e. for $\delta_1$ and $\delta_2$ we would have $\rho > 0.7$, which is quite large and undesirable.

To test our algorithms we generated random games of size $n\times n$, where each entry is picked uniformly at random from $[0, 1]$. The size of the games range from $500$ to $5000$ pure strategies with a step of $500$. For each size we generate 30 games and solve them to an accuracy of $\delta = 0.01$. 
We used two types of initialization in all methods, the fully uniform strategy profile and the profile  $(\ve_1, \ve_1)$, i.e., first row, first column. The latter has the advantage of not being too close to a Nash equilibrium from the start, in almost all games, and reveals more clearly the exploration that the method performs. 

The averaged results are presented in \Cref{fig:variant_comp}, where we show both the actual time and the number of iterations. In terms of actual time, our Fixed Support variant is the clear winner. Although \Cref{fig:variant_comp} reveals that as $n$ grows, the Fixed $\rho$ variant attains a lower number of iterations, this does not translate into improved running time. The intuition for this is that as $n$ grows and $\rho$ remains constant, we expect a larger number of strategies to be $\rho$-best responses. Consequently, the LP in \Cref{alg:find_dir} is closer to the full LP and thus more informative, but at the same time more expensive to solve. 
\begin{figure}
    \centering
    \includegraphics[width=\linewidth]{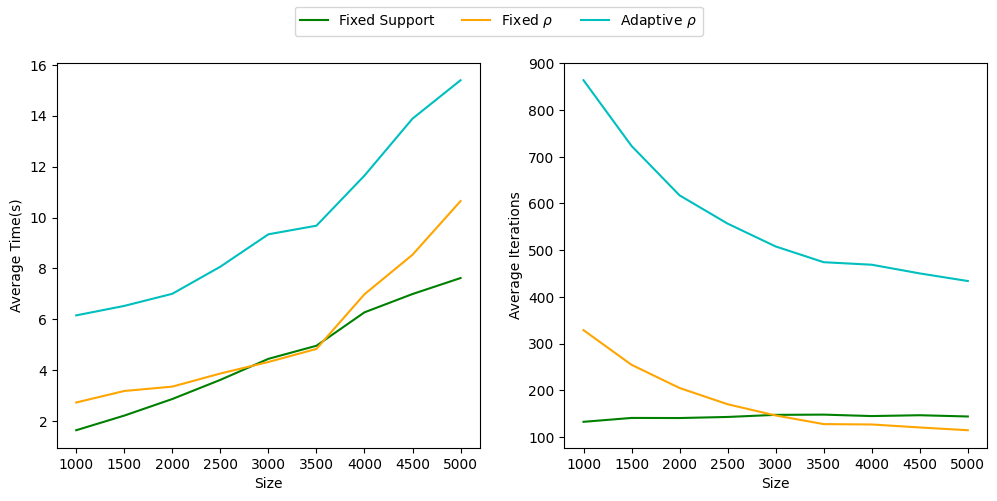}
    \caption{Average time and number of iterations for our variants}
    \label{fig:variant_comp}
\end{figure}

As a result of these comparisons, we select our Fixed Support variant as the variant to compare against other methods from the literature in the next subsection.

\subsection{Comparisons with LP and Gradient Methods}\label{subsec:comps}

We compared our Fixed Support variant against solving directly the full LP with a standard LP solver, and against a prominent first order method. Regarding the LP solver, we used the standard method of SciPy. We note that we used the same method for the smaller LPs that we solve in \Cref{alg:find_dir} of our methods.
To maintain an equal comparison with our algorithms, we used a tolerance of $0.01$. 
As for first order methods, we used Optimistic Gradient Descent Ascent (OGDA), which is among the fastest gradient based methods, with step size $\eta = 0.01$. We refer to \Cref{app:OGDA} for its definition. Another popular method is Optimistic Multiplicative Weights Update (OMWU), which however does not behave as well in practice, as also explained in \cite{CF+24}.

For each value of $n$ that we used, we generated  50 uniformly random games and 50 games using the Gaussian distribution. We also generated more structured but still random games, such as games with low rank. We present here the comparisons for the uniformly random games and we refer to \Cref{app:gauss,,app:low-rank} for the other classes of games. 
As in \Cref{sec:exp-ours}, we used two different initializations: starting from $(\ve_1, \ve_1)$ and starting from the uniform strategy profile: $(\frac{1}{n}, \dots, \frac{1}{n})$. 
The average running time can be seen in \Cref{fig:lpogda}. We summarize our findings as follows:
\begin{itemize}[leftmargin=*]
    \item The LP solver was far slower, even for lower values of $n$, as shown in the left subplot, and we dropped it from the experiments with larger games.
    \item When the initialization is $(e_1, e_1)$ (or any pure strategy profile), the advantage of our method is more clear (see left subplot of \Cref{fig:lpogda}). When we start with the uniform profile, we observe that our method is slower for smaller games but becomes faster in very large games (right subplot).
    \item Another observation is that our method seems smoother with less sharp jumps than OGDA when starting from $(\ve_1, \ve_1)$ while the opposite holds for the uniform profile. 

\end{itemize}

We view as the main takeaway of our experiments that our method is comparable to OGDA and in several cases even outperforms OGDA. One limitation of our current implementation is the choice of $\delta=0.01$. For much lower accuracies, our methods occasionally get stuck. We therefore feel that the overall approach deserves further exploration, especially on potential ways of accelerating its execution. 

\begin{figure}
    \centering
\includegraphics[width=\linewidth]{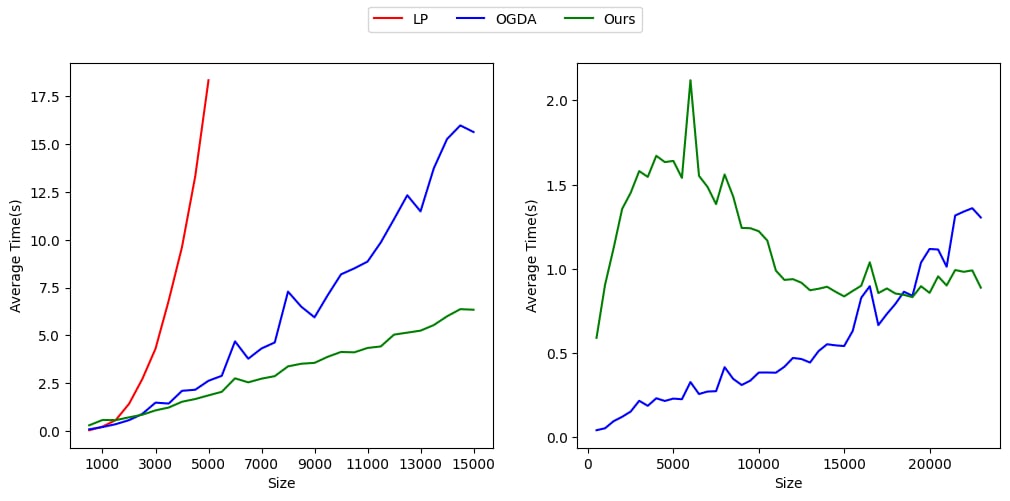}
    \caption{Time comparison between our Fixed Support Variant, LP solver and Optimistic Gradient Descent-Ascent}
    \label{fig:lpogda}
\end{figure}

\section{Conclusions}
\label{sec:conclusions}
We have analyzed a descent-based method for the duality gap in zero-sum games. Our goal has been to demonstrate the potential of such algorithms as a proof of concept. We expect that our method can be further optimized in practice and find this a promising direction for future work. In particular, one idea to explore is whether we can reuse the LP solutions we get in \Cref{alg:find_dir} from one iteration to the next (since we only change the current solution slightly by a step of size $\epsilon$). Exploring such {\it warm start} strategies (see e.g. \cite{YW02}) could provide significant speedups.

\bibliography{arxiv}

\newpage

\appendix

{\hypersetup{linkcolor=black}
\section{Missing Proofs from \Cref{sec:alg}}\label{app:missing-main}
}

{\hypersetup{linkcolor=black}
\subsection{Proof of \Cref{thm:comb-version}}\label{app:dr-form}
}

Using the definition of the duality gap, the directional derivative is equal to
\[
\nabla_{(\vx^\prime, \vy^\prime)}V(\vx, \vy)= 
\lim_{\varepsilon \to 0} 
\frac{\max_i \ve_i^\top \mR\Big((1-\varepsilon )\cdot \vy +\varepsilon \cdot \vy^\prime\Big)  - \min_j \Big((1-\varepsilon) \cdot \vx +\varepsilon\cdot \vx^\prime\Big)^\top \mR\ve_j - V(\vx, \vy)}{\varepsilon}.
\]

However, we can write up the term $\max_i \ve_i^\top \mR\Big((1-\varepsilon)\cdot \vy +\varepsilon \cdot \vy^\prime\Big)$, similarly to \cite{DFSS17,TS08}, as 
\begin{equation}
\label{Maximum analysis}
\begin{split}
&\max_i \ve_i^\top \mR\Big((1-\varepsilon) \cdot \vy +\varepsilon \cdot \vy^\prime\Big) = \max_{i\in BR_r(\vy)} \ve_i^\top \mR\Big((1-\varepsilon) \cdot  \vy +\varepsilon \cdot \vy^\prime\Big) \\ 
&+ \max \Big\{0, \max_{i\in \overline{BR_r(\vy)}} \ve_i^\top \mR\Big((1-\varepsilon)\cdot \vy +\varepsilon \cdot \vy^\prime\Big)-\max_{i\in BR_r(\vy)} \ve_i^\top \mR\Big((1-\varepsilon) \cdot \vy +\varepsilon \cdot  \vy^\prime\Big) \Big\},
\end{split}
\end{equation}

where $\overline{BR_r(\vy)}$ is the complement set of ${BR_r(\vy)}$.
Similarly,
\begin{equation}
\label{Minimum analysis}
\begin{split}
&\min_j \Big((1-\varepsilon)\cdot \vx +\varepsilon \cdot \vx^\prime\Big)^\top  R\ve_j = \min_{j\in BR_c(\vx)}\Big((1-\varepsilon)\cdot \vx +\varepsilon \cdot \vx^\prime\Big)^\top  R\ve_j \\ 
&- \max \Big\{0, -\min_{j\in \overline{BR_c(\vx)}}\Big((1-\varepsilon)\cdot \vx +\varepsilon \cdot \vx^\prime\Big)^\top  R\ve_j + \min_{j\in BR_c(\vx)} \Big((1-\varepsilon)\cdot \vx +\varepsilon \cdot \vx^\prime\Big)^\top  R\ve_j \Big\}.
\end{split}
\end{equation}
Arguing in a similar fashion as in \cite{DFSS17}, there exists $\epsilon^*>0$ such that for $\epsilon\leq \epsilon^*$, the term $$\max \Big\{0, \max_{i\in \overline{BR_r(\vy)}} \ve_i^\top \mR\Big((1-\varepsilon)\cdot  \vy +\varepsilon \cdot \vy^\prime\Big)\\-\max_{i\in BR_r(\vy)} \ve_i^\top  R\Big((1-\varepsilon) \cdot \vy +\varepsilon \cdot \vy^\prime\Big) \Big\}$$ 
is zero, and hence it can be ignored when we take the limit of $\epsilon \rightarrow 0$. 
In the same manner, the corresponding term for the column player also becomes $0$. 
 
Note also that for any $i, j\in BR_r(\vy)$, we have that $\ve_i^\top \mR \vy = \ve_j^\top \mR \vy$, and hence the term $\max_{i\in BR_r(\vy)} \ve_i^\top \mR \vy$ is independent of the row we choose, thus $\max_{i\in BR_r(\vy)} \ve_i^\top \mR\Big((1-\varepsilon )\cdot \vy +\varepsilon \cdot \vy^\prime\Big) = (1-\varepsilon )\cdot \max_{i\in BR_r(\vy)} \ve_i^\top \mR\vy +\varepsilon \cdot \max_{i\in BR_r(\vy)} \ve_i^\top \mR\vy^\prime$, similar for the $\min$ part. 
By using this below, we conclude that the directional derivative equals to
\begin{align*}
&\lim_{\varepsilon \to 0} 
\frac{\max_{i\in BR_r(\vy)} \ve_i^\top \mR\Big((1-\varepsilon )\cdot \vy +\varepsilon \cdot \vy^\prime\Big)- \min_{j\in BR_c(\vx)} \Big((1-\varepsilon) \cdot \vx +\varepsilon\cdot \vx^\prime\Big)^\top \mR\ve_j - V(\vx, \vy)}{\varepsilon} \\
=&\lim_{\varepsilon \to 0} 
\frac{(1-\varepsilon )\cdot\max_{i\in BR_r(\vy)}\ve_i^\top \mR \vy +\varepsilon \cdot \max_{i\in BR_r(\vy)}\ve_i^\top \mR \vy^\prime}{\varepsilon} \\ 
&- \frac{(1-\varepsilon) \cdot\min_{j\in BR_c(\vx)} x^\top \mR\ve_j +\varepsilon\cdot \min_{j\in BR_c(\vx)} (\vx^\prime)^\top \mR\ve_j + V(\vx, \vy)}{\varepsilon}\\
=&\lim_{\varepsilon \to 0} 
\frac{(1-\varepsilon )\cdot V(\vx, \vy) + \varepsilon \cdot \Big(\max_{i\in BR_r(\vy)} \ve_i^\top \mR \vy^\prime - \min_{j\in BR_c(\vx)} (\vx^\prime)^\top \mR\ve_j \Big) -V(\vx, \vy)}{\varepsilon} \\ 
= &\max_{i\in BR_r(\vy)}\ve_i^\top \mR \vy^\prime- \min_{j\in BR_c(\vx)} (\vx^\prime)^\top \mR\ve_j-V(\vx, \vy).
\end{align*}
as claimed.\\

{\hypersetup{linkcolor=black}
 \subsection{Proof of \Cref{lem:Bound approximation and exact derivative}}}
By definition, we have that
\begin{align*}
\nabla_{\vz^\prime}V(\vz)
 &= \max_{i\in BR_r(\vy)} \ve_i^\top \mR\vy^\prime-\min_{j\in BR_c(\vx )}(\vx^\prime)^\top \mR\ve_j - V(\vz)\\
 &\leq \max_{i\in BR^{\rho}_r(\vy)} \ve_i^\top \mR\vy^\prime-\min_{j\in BR^{\rho}_c(\vx )}(\vx^\prime)^\top \mR\ve_j - V(\vz)\\
 &=\nabla_{\rho,\vz^\prime}V(\vz),
 \end{align*}
 where the first inequality holds since, by definition, $BR_r(\vy) \subseteq BR^{\rho}_r(\vy)$ and $BR_c(\vx) \subseteq BR^{\rho}_c(\vx)$.\\

{\hypersetup{linkcolor=black}
\subsection{Proof of \Cref{lem:bound-approx_DD}}
}
Let $(\vx^*,\vy^*)$ be a Nash equilibrium, then it holds $V(\vx^*,\vy^*) = 0$ by the definition of the NE, but since the values of $V(\vx, \vy) \in [0,2]$ this implies that $(\vx^*,\vy^*)$ is a global minimum of the function in its domain. Let now a strategy profile $(\vx, \vy)$ such that $V(\vx, \vy)=0 = f_\mR(\vx, \vy) +f_{-\mR}(\vx, \vy)$, this trivially implies that $f_\mR(\vx, \vy) = 0$ and $f_{-\mR}(\vx, \vy) = 0$ since $f_\mR,f_{-\mR} \in [0,1]$, thus we have that $(\vx, \vy)$ is a NE in the zero-sum game.\\

{\hypersetup{linkcolor=black} \subsection{Proof of \Cref{lem:zero}}}
Firstly, we have that
\[
\max_{i\in BR^{\rho}_r(\vy)} \ve_i^\top \mR\Big((1-\varepsilon) \cdot \vy +\varepsilon \cdot  \vy^\prime\Big) \geq \max_{i\in BR^{\rho}_r(\vy)} \ve_i^\top \mR\Big((1-\varepsilon) \cdot \vy \Big)
= \max_{i\in BR^{\rho}_r(\vy)} \ve_i^\top \mR \vy - \varepsilon \cdot \max_{i\in BR^{\rho}_r(\vy)} \ve_i^\top \mR \vy. 
\]
By the definition of the $\max$ function, we have
\begin{equation*}
\begin{split}
 \max_{i\in \overline{BR^{\rho}_r(\vy)}} \ve_i^\top \mR\Big((1-\varepsilon)\cdot \vy +\varepsilon \cdot \vy^\prime\Big)\leq (1-\varepsilon) \cdot \max_{i\in \overline{BR^{\rho}_r(\vy)}} \ve_i^\top \mR \vy + \varepsilon \cdot \max_{i\in \overline{BR^{\rho}_r(\vy)}} \ve_i^\top \mR \vy'.
\end{split}
\end{equation*}
These two bounds give
\begin{gather*}
 \max_{i\in \overline{BR^{\rho}_r(\vy)}} \ve_i^\top \mR\Big((1-\varepsilon)\cdot \vy +\varepsilon \cdot \vy^\prime\Big)-\max_{i\in BR^{\rho}_r(\vy)} \ve_i^\top \mR\Big((1-\varepsilon) \cdot \vy +\varepsilon \cdot  \vy^\prime\Big)\\
 \leq \max_{i\in \overline{BR^{\rho}_r(\vy)}} \ve_i^\top \mR \vy + \varepsilon \cdot \Big (\max_{i\in \overline{BR^{\rho}_r(\vy)}} \ve_i^\top \mR \vy' - \max_{i\in \overline{BR^{\rho}_r(\vy)}} \ve_i^\top \mR \vy \Big)
 -\max_{i\in BR^{\rho}_r(\vy)} \ve_i^\top \mR \vy + \varepsilon \cdot \max_{i\in BR^{\rho}_r(\vy)} \ve_i^\top \mR \vy\\
 =\max_{i\in \overline{BR^{\rho}_r(\vy)}} \ve_i^\top \mR \vy-\max_{i\in BR^{\rho}_r(\vy)} \ve_i^\top \mR \vy + \varepsilon \cdot \Big (\max_{i\in \overline{BR^{\rho}_r(\vy)}} \ve_i^\top \mR \vy' - \max_{i\in \overline{BR^{\rho}_r(\vy)}} \ve_i^\top \mR \vy 
  + \max_{i\in BR^{\rho}_r(\vy)} \ve_i^\top \mR \vy\Big).
\end{gather*}
By the definition of $\rho$-best-response, we have that

\begin{equation*}
\begin{split}
\max_{i\in \overline{BR^{\rho}_r(\vy)}} \ve_i^\top \mR \vy-\max_{i\in BR^{\rho}_r(\vy)} \ve_i^\top \mR \vy < -\rho.
\end{split}
\end{equation*}
Furthermore, we have that $\max_{i\in \overline{BR^{\rho}_r(\vy)}} \ve_i^\top \mR \vy' - \max_{i\in \overline{BR^{\rho}_r(\vy)}} \ve_i^\top \mR \vy 
  + \max_{i\in BR^{\rho}_r(\vy)} \ve_i^\top \mR \vy \leq 2$, since $R_{ij} \leq 1$.
Thus, we have that
\begin{equation*}
\begin{split}
\max_{i\in \overline{BR^{\rho}_r(\vy)}} \ve_i^\top \mR \vy-\max_{i\in BR^{\rho}_r(\vy)} \ve_i^\top \mR \vy + \varepsilon \cdot \Big (\max_{i\in \overline{BR^{\rho}_r(\vy)}} \ve_i^\top \mR \vy' - \max_{i\in \overline{BR^{\rho}_r(\vy)}} \ve_i^\top \mR \vy 
  + \max_{i\in BR^{\rho}_r(\vy)} \ve_i^\top \mR \vy\Big)
  < -\rho + 2 \cdot \varepsilon.
\end{split}
\end{equation*}
So, we want to find a value of $\varepsilon$ such that $-\rho + 2\cdot \varepsilon \leq 0$, which holds for $\varepsilon \leq \frac{\rho}{2}$.
In a very similar fashion, for the second part of the Lemma, we have

\begin{equation*}
\begin{split}
\min_{j\in \overline{BR^{\rho}_c(\vx)}}\Big((1-\varepsilon)\cdot \vx +\varepsilon \cdot \vx^\prime\Big)^\top  R\ve_j 
\geq (1-\varepsilon) \cdot \min_{j\in \overline{BR^{\rho}_c(\vx)}}\vx^\top  R\ve_j + \varepsilon \cdot \min_{j\in \overline{BR^{\rho}_c(\vx)}}(\vx^\prime)^\top  R\ve_j.
\end{split}
\end{equation*}
Furthermore,
\begin{align*}
 \min_{j\in BR^{\rho}_c(\vx)} \Big((1-\varepsilon)\cdot \vx +\varepsilon \cdot \vx^\prime\Big)^\top  R\ve_j
 &= - \max_{j\in BR^{\rho}_c(\vx)} \Big((1-\varepsilon)\cdot \vx +\varepsilon \cdot \vx^\prime\Big)^\top (-R)\ve_j\\
 &\leq -(1-\varepsilon) \cdot \max_{j\in BR^{\rho}_c(\vx)} \vx^\top (-R)\ve_j = (1-\varepsilon) \cdot \min_{j\in BR^{\rho}_c(\vx)} \vx^\top R\ve_j.\\
\end{align*}
The above inequality holds since $\max_{j} \Big((1-\varepsilon)\cdot \vx +\varepsilon \cdot \vx^\prime\Big)^\top (-R)\ve_j \geq  \max_{j} \Big((1-\varepsilon)\cdot \vx\Big)^\top(-R)\ve_j = (1-\varepsilon) \cdot \max_{j\in BR^{\rho}_c(\vx)} \vx^\top  (-R) \ve_j$. Thus, these two bounds give
\begin{gather*}
-\min_{j\in \overline{BR^{\rho}_c(\vx)}}\Big((1-\varepsilon)\cdot \vx +\varepsilon \cdot \vx^\prime\Big)^\top  R\ve_j + \min_{j\in BR^{\rho}_c(\vx)} \Big((1-\varepsilon)\cdot \vx +\varepsilon \cdot \vx^\prime\Big)^\top  R\ve_j\\
\leq  -(1-\varepsilon) \cdot \min_{j\in \overline{BR^{\rho}_c(\vx)}}\vx^\top  R\ve_j - \varepsilon \cdot \min_{j\in \overline{BR^{\rho}_c(\vx)}}(\vx^\prime)^\top  R\ve_j
+(1-\varepsilon) \cdot \min_{j\in {BR^{\rho}_c(\vx)}}\vx^\top  R\ve_j \\
=-\min_{j\in \overline{BR^{\rho}_c(\vx)}}\vx^\top  R\ve_j+\min_{j\in {BR^{\rho}_c(\vx)}}\vx^\top  R\ve_j
+\varepsilon \cdot \Big(\min_{j\in \overline{BR^{\rho}_c(\vx)}}\vx^\top  R\ve_j-\min_{j\in \overline{BR^{\rho}_c(\vx)}}(\vx^\prime)^\top  R\ve_j -\min_{j\in {BR^{\rho}_c(\vx)}}\vx^\top  R\ve_j \Big).
\end{gather*}

But, by the definition of $\rho$-best-response, we have that
\begin{equation*}
\begin{split}
-\min_{j\in \overline{BR^{\rho}_c(\vx)}}\vx^\top  R\ve_j+\min_{j\in {BR^{\rho}_c(\vx)}}\vx^\top  R\ve_j < -\rho,
\end{split}
\end{equation*}
and $$\min_{j\in \overline{BR^{\rho}_c(\vx)}}\vx^\top  R\ve_j-\min_{j\in \overline{BR^{\rho}_c(\vx)}}(\vx^\prime)^\top  R\ve_j -\min_{j\in {BR^{\rho}_c(\vx)}}\vx^\top  R\ve_j \leq 1,$$
since $R_{ij} \leq 1$. In total, we have that 
\begin{equation*}
\begin{split}
\min_{j\in {BR^{\rho}_c(\vx)}}\!\vx^\top  R\ve_j-\min_{j\in \overline{BR^{\rho}_c(\vx)}}\!\vx^\top  R\ve_j
+\varepsilon \Big(\!\min_{j\in \overline{BR^{\rho}_c(\vx)}}\vx^\top  R\ve_j-\min_{j\in \overline{BR^{\rho}_c(\vx)}}(\vx^\prime)^\top  R\ve_j -\min_{j\in {BR^{\rho}_c(\vx)}}\vx^\top  R\ve_j \Big)
< -\rho + \varepsilon.
\end{split}
\end{equation*}
Thus, we only need to ensure that $-\rho + \varepsilon \leq  0$, which is true when $\varepsilon \leq \frac{\rho}{2}$.\\
    
{\hypersetup{linkcolor=black}
\section{Equivalent Formulation of \Cref{alg:find_dir}}\label{app:find_dir}
}
Recall that the linear program involved in Find Direction is the following, given a profile $(\vx, \vy)$: 
\begin{align*}
    \min\quad &\gamma \\
    \text{s.t.}\quad &\gamma \geq (\ve_i)^\top \mR\vy^\prime -(\vx^\prime)^\top \mR\ve_j  \\
    &i \in BR^{\rho}_r(\vy), j \in BR^{\rho}_c(\vx), \vx^\prime \in \Delta^{n-1}, \vy^\prime \in \Delta^{n-1} 
\end{align*}

It is easy to see that this is equivalent to:
\begin{align*}
    \min\quad &\gamma_1 + \gamma_2 \\
    \text{s.t.}\quad &\gamma_1 \geq (\ve_i)^\top \mR\vy^\prime  \\
    &\gamma_2 \ge -(\vx^\prime)^\top \mR\ve_j \\
    &i \in BR^{\rho}_r(\vy), j \in BR^{\rho}_c(\vx), \vx^\prime \in \Delta^{n-1}, \vy^\prime \in \Delta^{n-1} 
\end{align*}
Now, we see that the variables of $\vx^\prime$ appear in separate constraints from the variables of $\vy^\prime$. Hence, it suffices to solve separately the problems:

\begin{minipage}{.45\textwidth}
    \begin{align*}
    \min\quad &\gamma_1\\
    \text{s.t.}\quad &\gamma_1 \geq (\ve_i)^\top \mR\vy^\prime  \\
    &i \in BR^{\rho}_r(\vy), \vy^\prime \in \Delta^{n-1}
    \end{align*}
\end{minipage}
\begin{minipage}{.45\textwidth}
    \begin{align*}
    \min\quad &\gamma_2\\
    \text{s.t.}\quad &\gamma_2 \geq -(\vx^\prime)^\top \mR\ve_j  \\
    &j \in BR^{\rho}_c(\vx), \vx^\prime \in \Delta^{n-1}
    \end{align*}
\end{minipage}\vspace{.2cm}

With this formulation at hand we can now prove \Cref{obs:rho=1}.

\begin{proofof}[Proof of {\hypersetup{linkcolor=black}\Cref{obs:rho=1}}]

In the case where $\rho=1$ we have that $BR^{\rho}_c(\vx) =[n]$ and $BR^{\rho}_r(\vy) = [n]$. This means that there is a constraint in each LP for every pure strategy, and this means that the LPs reduce to the well known primal and dual LP formulation of computing the equilibrium strategies in zero-sum games. Note that the second LP can be written as a maximization problem (for $-\gamma_2$), which makes the primal-dual connection more apparent.\end{proofof}

\section{OGDA and Other First Order Methods}\label{app:OGDA}
The equations for the Optimistic Gradient Descent/Ascent are as follows:
\begin{align*}
\vx_{t+1} &= \vx_t - 2\alpha \nabla_{\vx} f(\vx_t,\vy_t) + \alpha \nabla_{\vx} f(\vx_{t-1},\vy_{t-1})\\
\vy_{t+1} &= \vy_t - 2\alpha \nabla_{\vy} f(\vx_t,\vy_t) + \alpha \nabla_{\vy} f(\vx_{t-1},\vy_{t-1})
\end{align*}

\noindent where for bilinear objective functions, such as zero-sum games we have $f(\vx,\vy) = x^\top \mR y$. \\

To serve as a means of comparison to our method, we implemented the OGDA algorithm, following the aforementioned equations. To verify the validity of our implementation, we present a comparison with the corresponding PyOL(Python library for Online Learning) version of the same algorithm. In the figure that follows, the algorithms are set with $\eta = 0.01$ and are tested on uniformly random games.

\begin{figure}[H]
    \centering
    \includegraphics[width=0.5\linewidth]{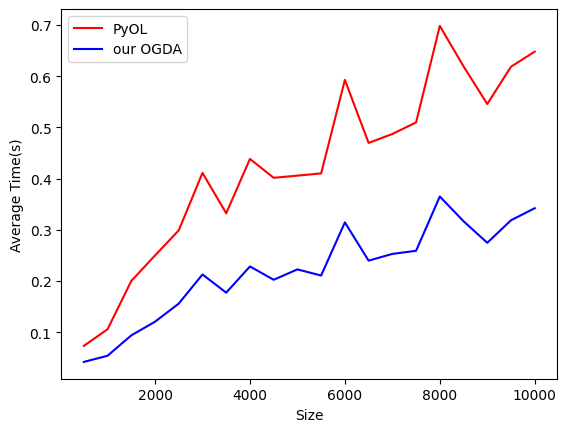}
    \caption{Time comparison between PyOL and our OGDA implementations}
    \label{fig:pyol}
\end{figure}

There exist also other first order methods, which have good theoretical guarantees. A popular such method is 
Optimistic Multiplicative Weights Update (OMWU), an optimistic variant of MWU. 
We realized however that MWU is less stable than OGDA and exhibits cyclic behavior in certain cases. 
In fact this has been also observed by other works and the recent paper by \cite{CF+24} provides further justification 
on why some first order methods may fail to be a good practical solution.

As a conclusion, we decided to present comparisons of our method only against OGDA in \Cref{sec:experiments}.

\section{Additional Experiments}\label{app:exp}
Here we present additional experiments regarding the choice of parameters and more comparisons between our variants and other methods. All the experiments here and in \Cref{sec:experiments} were run on a Macbook M1 Pro(10 core) with 16GB RAM.  We developed our code in Python 3.10.9, using the packages NumPy 2.0.2 and SciPy 1.14.1. 

\subsection{Comparison with Constant Step Size $\epsilon$}\label{app:eta}

We begin our experiments with the comparison between our Fixed Support Variant (with optimized $\varepsilon$) and the Fixed Support Variant with a constant $\varepsilon$, along the spirit of standard first-order methods using a constant learning rate. We present two comparisons in \Cref{fig:constanteta} for a random game of a size $n=5000$. In the left subplot we see the result of starting from $(\ve_1, \ve_1)$. As expected, the optimized version here is much faster, since when the duality gap is large there is a room for big steps. In addition, we also see (in the right subplot) that even for the uniform initialization, i.e. when we are close to an approximate Nash equilibrium, that the optimized version terminates faster. Given that the iterations are more expensive, we should compare the times as well. Indeed, the constant step size requires 1.8 seconds vs 1.15 for the optimized version.

Therefore, we selected to use our optimized version for the choice of $\epsilon$ as described in \Cref{sec:theory-to-impl}.

\begin{figure}[H]
    \centering
    \includegraphics[width=0.5\linewidth]{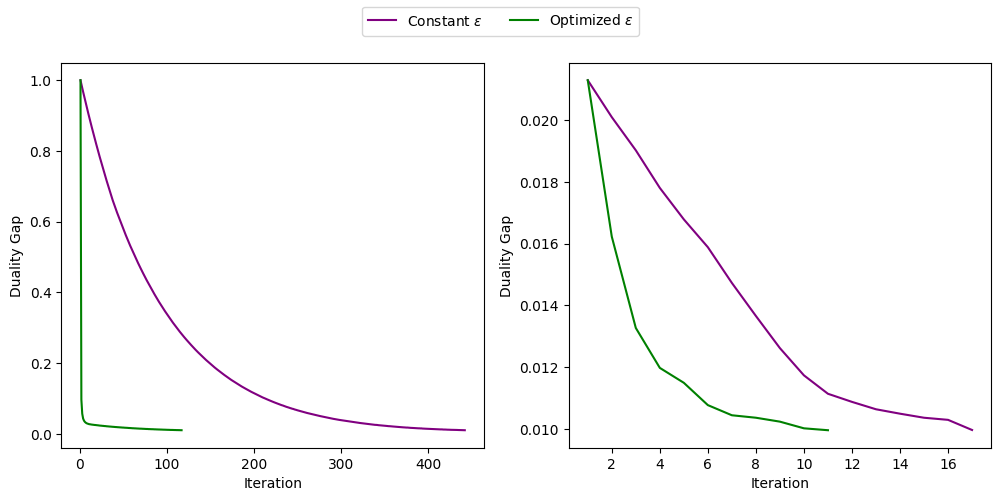}
    \caption{Iteration comparison between for the initializations $(\ve_1, \ve_1)$ (left) and the uniform point (right)}
    \label{fig:constanteta}
\end{figure}
\subsection{Tuning our Fixed Support Variant}\label{app:k}

For the use of our Fixed Support variant, recall that we used a parameter $k$ to denote the number of the top better responses that we select to include when constructing the LP constraints in \Cref{alg:find_dir}.
We tried to experiment with the appropriate size for $k$ (also referred to as the support size).
We present in \Cref{fig:ktesting} the averaged results for 20 games of size $n=10000$. As the figure reveals, the range of $k$ between $[60, 130]$ seemed to behave quite well, and as a result, we selected $k=100$. We would like to point out that the image is similar for other game sizes.

\begin{figure}[H]
    \centering
    \includegraphics[width=0.5\linewidth]{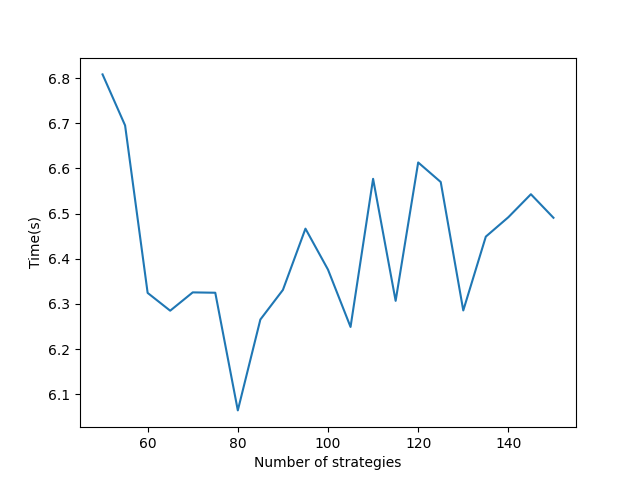}
    \caption{Running time comparison for different values of $k$}
    \label{fig:ktesting}
\end{figure}

We propose, as an interesting direction for future work, some possible modifications. One is to set $k$ as a function of $n$, i.e. $k = \log n$. Having the support size growing with $n$ is natural and we cannot rule out the possibility that for even larger games sizes, with size in the millions, $k=100$ would be too small to achieve high performance. Another approach, quite similar to one approach with $\rho$, would be to use an adaptive support size. In this case, the adaptation could be over the step size $\varepsilon$; whenever the step becomes too small, we can increase $k$ and try for larger steps in the next iterations.  
\subsection{Gaussian Random Games}\label{app:gauss}

Our next experiment concerns a different class of randomly generated games: instead of sampling from the uniform distribution we sample from the standard Gaussian, and then rescale to have the entries back in $[0,1]$. Qualitatively, we observed no difference between our methods; the Fixed Support variant, with the same parameters as in the uniform game generation, was the winner among all our variants, as before. Therefore, here we only show a comparison between our Fixed Support variant and OGDA, in \Cref{fig:lpogda_gauss}. Our findings are consistent with what we presented in \Cref{sec:experiments}. Additionally here we notice that both methods are faster, compared to the times with the uniform distribution. Our method matches and outperforms OGDA at an earlier range of game sizes (at around $n = 5000$ as opposed to $n=18000$ for the uniformly random games). 

\begin{figure}[H]
    \centering
    \includegraphics[width=0.5\linewidth]{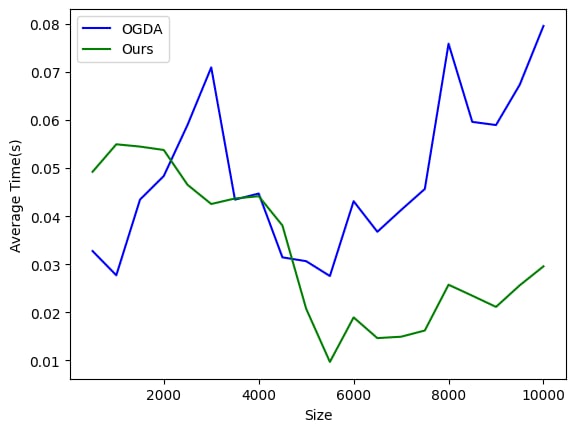}
    \caption{Comparison with Gaussian Random Games}
    \label{fig:lpogda_gauss}
\end{figure}

\subsection{Low Rank Zero-sum Games}\label{app:low-rank}

We also experimented with a class of games where the payoff entries are not all drawn independently from each other, which translates to a lower than full rank payoff matrix.
To generate a matrix $\mR$ of size $n$ and fixed rank $r$, we uniformly sample matrices $\mU, \mV$ of dimensions $n \times r$. Obtaining $\mR$ via $\mR = \mU\cdot \mV^\top$ gives us a matrix of desired rank $r$. 

This class of games provided new insights for our method. First, we noticed that both our Fixed Support Variant and OGDA were getting stuck and/or significantly slowed down to lower accuracies. To amend this issue for OGDA, we experimented with different values of constant learning rates but this was to no avail. Then, we turned to the square rooted learning rate option of PyOL library which helped accelerate OGDA to some degree. 
On the other hand, for our method there was a parameter we had not paid much attention to earlier, that needed tweaking: the tolerance parameter of the LP solver in \Cref{alg:find_dir}. Once using a lower tolerance, i.e. asking for higher accuracy solutions to our small LPs the speed of our method improves drastically! 
The result for matrices of rank 10 is depicted in \Cref{fig:low-rank} with the picture being similar for higher ranks as well. In particular, we observe that we outperform OGDA as soon as the game size exceed 1000 strategies.

\begin{figure}[H]
    \centering
    \includegraphics[width=0.5\linewidth]{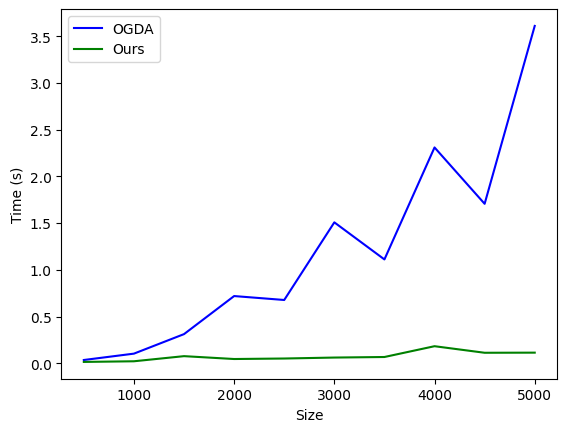}
    \caption{Comparison for fixed rank $r=10$}
    \label{fig:low-rank}
\end{figure}

\end{document}

%% file: math_commands.tex
\usepackage{amsmath,amsfonts,bm}

\def\eqref#1{equation~\ref{#1}}
\def\1{\bm{1}}

\def\ve{{\bm{e}}}

\def\vx{{\bm{x}}}
\def\vy{{\bm{y}}}
\def\vz{{\bm{z}}}

\def\mR{{\bm{R}}}

\def\mU{{\bm{U}}}
\def\mV{{\bm{V}}}

\DeclareMathAlphabet{\mathsfit}{\encodingdefault}{\sfdefault}{m}{sl}
\SetMathAlphabet{\mathsfit}{bold}{\encodingdefault}{\sfdefault}{bx}{n}